\documentclass[11pt]{llncs}
\usepackage{llncsdoc}
\usepackage{etex}
\usepackage{amsmath,epsfig,amsfonts}
\usepackage{fullpage}
\usepackage{array}
\usepackage{listings}
\usepackage{times}
\usepackage[ruled,vlined,linesnumbered]{algorithm2e}
\usepackage{qtree}
\usepackage{subfigure}
\usepackage{hyperref}

\usepackage{caption}

\usepackage[usenames,dvipsnames]{pstricks}
\usepackage{epsfig}
\usepackage{pst-grad} 
\usepackage{pst-plot} 

\usepackage{tikz}

\newtheorem{fact}{Fact}

\newcommand{\qedsymb}{\hfill{\rule{2mm}{2mm}}}

\newcommand{\remove}[1]{}

\DeclareMathOperator{\sus}{\mathit{SUS}}
\DeclareMathOperator{\lsus}{\mathit{LSUS}}
\DeclareMathOperator{\sa}{\mathit{SA}}
\DeclareMathOperator{\rank}{\mathit{RA}}
\DeclareMathOperator{\lcp}{\mathit{LCP}}
\DeclareMathOperator{\sls}{\mathit{SLS}}

\DeclareMathOperator{\llr}{\mathit{LLR}}

\DeclareMathOperator{\nil}{\mathit{NIL}}
\DeclareMathOperator{\seq}{\mathit{S}}

\newcommand\pred{\mathtt{pred}}

\title{ 
  An In-place Framework for Exact and Approximate\\Shortest Unique
  Substring Queries \thanks{Authors are listed in alphabetical order.
    A preliminary version of this paper appears in Proceedings of the
    26th International Symposium on Algorithms and Computation
    (ISAAC),  Nagoya, Japan, 2015.
}
}

\author{Wing-Kai Hon\inst{1} \and Sharma V.\ Thankachan\inst{2} \and
  Bojian Xu\inst{3}
\thanks{Corresponding author. Phone:
+1 (509) 359-2817. Fax: +1 (509) 359-2215.}  
}

\institute{
Department of CS, National Tsing Hua University,
Taiwan
\and
School of CSE, Georgia Institute of Technology, 
USA
\and
Department of CS, Eastern Washington University, 
USA\\
\email{wkhon@cs.nthu.edu.tw, sthankac@cc.gatech.edu, bojianxu@ewu.edu}
}

\begin{document}
\maketitle
\begin{abstract}
  We revisit the exact shortest unique substring (SUS) finding
  problem,
  and propose its approximate version where mismatches are allowed,
  due to its applications in subfields such as computational
  biology. We design a generic in-place framework that fits to solve
  both the exact and approximate $k$-mismatch SUS finding, using the
  minimum $2n$ memory words plus $n$ bytes space, where  $n$ is the input
  string size.
  By using the in-place framework, we can find the exact and
  approximate $k$-mismatch SUS for every string position using a total
  of $O(n)$ and $O(n^2)$ time, respectively, regardless of the value
  of $k$.
  Our framework does not involve any compressed or succinct data
  structures and thus is practical and easy to implement.
  \remove{ 
   Unlike prior work that used either suffix tree or suffix array
    augmented by longest common prefix array, our method uses suffix
    array only for the exact SUS finding and does not use any indexing
    structures for the approximate SUS finding, making our solution
    practical and easy to implement.
}
 \end{abstract}


\section{Introduction}
\label{sec:intro}
We consider a {\bf string} $S[1 .. n]$, where
each character $S[i]$ is drawn from an
alphabet $\Sigma=\{1,2,\ldots, \sigma\}$. 
We say the character $S[i]$ {\bf occupies} the string position $i$.
A {\bf substring} $S[i.. j]$ of $S$ represents $S[i]S[i+1]\ldots S[j]$
if $1\leq i\leq j \leq n$, and is an empty string if $i>j$.  
We call $i$ the {\bf start position} and $j$ the 
{\bf ending position} of $S[i..j]$.  
We say the substring $S[i.. j]$ {\bf covers} the $k$th
position of $S$, if $i\leq k \leq j$.  String $S[i'.. j']$ is a {\bf
  proper substring} of another string $S[i.. j]$ if $i\leq i' \leq j'
\leq j$ and $j'-i' < j-i$.
The {\bf length} of a non-empty substring $S[i.. j]$, denoted as
$|S[i.. j]|$, is $j-i+1$. We define the length of an empty string
as zero. 

The {\bf Hamming distance} of two non-empty strings $A$ and $B$ of
equal length, denoted as $H(A,B)$, is defined as the number of string
positions where the characters differ.  A substring $S[i .. j]$
is {\bf $k$-mismatch unique}, for some $k\geq 0$, if there does
not exist another substring $S[i'.. j']$, such that $i'\neq i$,
$j-i = j'-i'$, and $H(S[i..j],S[i'..j'])\leq k$.  A substring is a
{\bf $k$-mismatch repeat} if it is not $k$-mismatch unique.

\begin{definition}[$k$-mismatch SUS]
\label{def:sus}
For a particular string position $p$ in $S$ and an integer $k$, $0\leq
k \leq n-1$, the \emph{$k$-mismatch shortest unique substring (SUS)
  covering position $p$}, denoted as $\sus_p^k$, is a $k$-mismatch
unique substring $S[i..j]$, such that (1) $i\leq p \leq j$, and (2) there
does not exist another $k$-mismatch unique substring $S[i'..j']$, such
that $i'\leq p \leq j'$ and $j'-i' < j-i$.
\end{definition}

We call $0$-mismatch SUS as {\bf exact SUS},
and the case $k>0$ as {\bf approximate SUS}.

For any $k$ and $p$, $\sus_p^k$ must exist, because at least the
string $S$ can be $\sus_p^k$, if none of its proper substrings is
$\sus_p^k$.  On the other hand, there might be multiple choices for
$\sus_p^k$.  For example, if $S={\tt abcbb}$, $\sus_2^0$ can be
either $S[1,2]={\tt ab}$ or $S[2,3]={\tt bc}$, and $\sus_2^1$ can be
either $S[1..3]={\tt abc}$ or $S[2..4]={\tt bcb}$.
Note that in Definition~\ref{def:sus}, we require $k < n$, because
finding $\sus_p^n$ is trivial: $\sus_p^n\equiv S$ for any string
position $p$.

\bigskip  

\noindent
{\bf Problem ($k$-mismatch SUS finding).}
Given the string $S$, the value of $k\geq 0$, and two empty integer arrays
$A$ and $B$, we want to work in the place of $S$, $A$, and $B$, such
that, in the end of computation: (1) $S$ does not change. (2) Each
$(A[i],B[i])$ pair saves the start and ending positions of the
rightmost\footnote{Since any SUS may have multiple choices, it is our
  arbitrary decision to resolve the ties by picking the rightmost
  choice. However, our solution can also be easily modified to find
  the leftmost choice.}  $\sus_i^k$, i.e., 
$S\bigl[A[i].. B[i]\bigr] = \sus_i^k$,
 using a total of $O(n)$ time for
$k=0$ and $O(n^2)$ time for any $k\geq 1$.

\subsection{Prior work and our contribution}
Exact SUS finding was proposed and studied recently by Pei et
al.~\cite{PWY-ICDE2013}, due to its application in locating snippets 
in document
search, event analysis, and bioinformatics, such as
finding the distinctness
between closely related organisms~\cite{HPMW-bio2005}, polymerase chain reaction
(PCR) primer design in molecular biology, genome
mapability~\cite{DEMKRGR-2012}, and next-generation short reads
sequencing~\cite{ABFMK-2015}.  The algorithm in~\cite{PWY-ICDE2013} can find all
exact SUS in $O(n^2)$ time using a suffix tree of $O(n)$ space.
Following their proposal, there has been a sequence of
improvements~\cite{TIBT2014,IKX-tcs2015}
for exact SUS finding, reducing the time cost from $O(n^2)$ to $O(n)$
and alleviating the underlying data structure from suffix tree to
suffix array of $O(n)$ space. 
Hu et al.~\cite{HPT-spire2014} proposed an RMQ (range minimum query)
technique based indexing structure, which can be constructed in $O(n)$
time and space, such that any future exact SUS covering any interval of
string positions can be answered in $O(1)$ time.
In this work, we make the following contributions: 
\begin{itemize}
\item 
We revisit the exact SUS finding problem and also
propose its approximate version where mismatches are allowed, which
significantly increases the difficulty as well as
the usage of SUS finding in subfields such as
bioinformatics, where approximate string matching is unavoidable due
to genetic mutation and errors in biological experiments.

\item 
We propose a generic in-place algorithmic framework that fits to solve
both the exact and approximate $k$-mismatch SUS finding, using $2n$
words plus $n$ bytes space. It is worth mentioning that $2n$ words
plus $n$ bytes is the minimum memory space needed to save those $n$
calculated SUSes: (1) It needs $2$ words to save each SUS by saving
its start and ending positions (or one endpoint and its length) and
there are $n$ SUSes. (2) It needs another $n$ bytes to save the
original string $S$ in order to output the actual content of any SUS
of interest from queries. Note that all prior
work~\cite{PWY-ICDE2013,TIBT2014,IKX-tcs2015,HPT-spire2014} use $O(n)$
space but there is big leading constant hidden within the big-oh
notation (see the experimental study in~\cite{IKX-tcs2015}).

\item After the suffix array is constructed, all the computation in
  our solution happens in the place of two integer arrays, 
 using non-trivial techniques. It is worth noting that our solution
 does not 
  involve any compressed or succinct data structures, making our solution
  practical and easy to implement.  Our preliminary experimental study
  shows that our solution for exact SUS finding is even faster than
  the fastest one
  among~\cite{PWY-ICDE2013,TIBT2014,IKX-tcs2015}\footnote{Note
    that the work of~\cite{HPT-spire2014} studies a different problem and its
    computation is of the query-answer model, and thus is not
    comparable with~\cite{PWY-ICDE2013,TIBT2014,IKX-tcs2015} and
    ours.}, in addition to a lot more space saving than them, enabling
  our solution to handle larger data sets.  Due to page limit, we will
  deliver the details of our experimental study in the journal version
  of this paper.
\end{itemize}


\remove{

\begin{figure}[t]
  \centering

 \begin{tabular}{||l|l|l|l|l|l||}
\hline
          & $n$ & \cite{PWY-ICDE2013} & \cite{TIBT2014} &
          \cite{IKX-tcs2015} & Ours\\
\hline
\hline
  English & & & & &\\
 Protein & & & & &\\
 DNA & & & & &\\
\hline
\end{tabular}
  figure on time here ...

  \caption{Time and space usage of different proposals for exact SUS
    finding over different strings, where $n$ is the string size in
    MB.}
\label{fig:time-space}
\end{figure}

}

\section{Preparation}
\label{sec:prep}
A {\bf prefix} of $S$ is a substring $S[1.. i]$, $1\leq i\leq n$. 
A {\bf proper prefix} $S[1.. i]$ is a prefix of $S$ where $i <
n$.
A {\bf suffix} of $S$ is a substring
$S[i.. n]$, denoted as $S_i$, $1\leq i\leq n$.  
 $S_i$ is a  {\bf proper suffix} of $S$, if $i >
1$.


%
For two strings $A$ and $B$, we write ${\bf A=B}$ (and say $A$ is {\bf
  equal} to $B$), if $|A|= |B|$ and $H(A,B)=0$.  
%
We say $A$ is lexicographically smaller than $B$,
denoted as ${\bf A < B}$, if (1) $A$ is a proper prefix of $B$, or (2)
$A[1] < B[1]$, or (3) there exists an integer $k > 1$ such that
$A[i]=B[i]$ for all $1\leq i \leq k-1$ but $A[k] < B[k]$.

The {\bf suffix array} $\sa[1..n]$ of the string $S$ is a permutation
of $\{1,2,\ldots, n\}$, such that for any $i$ and $j$, $1\leq i < j
\leq n$, we have $S[\sa[i].. n] < S[\sa[j] .. n]$.  That is, $\sa[i]$
is the start position of the $i$th smallest suffix in the
lexicographic order.
The {\bf rank array} $\rank[1 .. n]$ is the inverse of the suffix
array, i.e., $\rank[i]=j$ iff $\sa[j]=i$.
The {\bf $k$-mismatch longest common prefix (LCP)} between two strings
$A$ and $B$, $k\geq 0$, denoted as $\lcp^k(A,B)$, is the LCP
of $A$ and $B$ within
Hamming distance $k$. For example, if $A={\tt abc}$ and $B={\tt
  acb}$,
then: 
$\lcp^0(A,B)$ is $A[1]=B[1] = {\tt a}$ and $|\lcp^0(A,B)|$ = 1;
$\lcp^1(A,B)$ is $A[1..2]={\tt ab}$ and $B[1,2]={\tt ac}$ and
$|\lcp^1(A,B)|=2$.

\begin{definition}[$k$-mismatch LSUS]
\label{def:lsus}
For a particular string position $p$ in $S$ and an integer $k$, $0\leq
k \leq n-1$, the \emph{$k$-mismatch left-bounded shortest unique
  substring (LSUS) starting at position $p$}, denoted as $\lsus_p^k$,
is a $k$-mismatch unique substring $S[p..  j]$, such that either $p=j$
or any proper prefix of $S[p.. j]$ is not $k$-mismatch unique.
\end{definition}

We call $0$-mismatch LSUS as {\bf exact
  LSUS}, and the case $k>0$ as {\bf approximate
  LSUS}.

Observe that for any $k$, $\lsus_1^k=\sus_1^k$ always exists, because
at least the whole string $S$ can be $\lsus_1^k$. However, for any
$k\geq 0$ and $p\geq 2$, $\lsus_p^k$ may not exist. For example, if
$S={\tt dabcabc}$, none of $\lsus_i^0$ and $\lsus_j^1$ exists, for all
$i\geq 5$, $j\geq 4$. It follows that some string positions may not be
covered by any $k$-mismatch LSUS. For example, for the same string
$S={\tt dabcabc}$, positions $6$ and $7$ are not covered by any exact
or $1$-mismatch LSUS. On the other hand, if any $\lsus_p^k$ does
exist, there must be only one choice for $\lsus_p^k$, because
$\lsus_p^k$ has its start position fixed on $p$ and need to be as
short as possible.
Note that in Definition~\ref{def:lsus}, we require $k < n$, because
finding $\lsus_p^n$ is trivial as $\lsus_1^n\equiv S$ and $\lsus_p^n$
does not exist for all $p>1$.

\begin{definition}[$k$-mismatch SLS]
\label{def:sls}
For a particular string position $p$ in $S$ and an integer $k$, $0\leq
k \leq n-1$, we use $\sls_p^k$ to denote the \emph{shortest $k$-mismatch LSUS 
covering position $p$}.
\end{definition}

We call $0$-mismatch SLS as {\bf exact
  SLS}, and the case $k>0$ as {\bf approximate
  SLS}.

$\sls_p^k$ may not exist, since position $p$ may not be covered by any
$k$-mismatch LSUS at all. For example, if $S={\tt dabcabc}$, then none
of $\sls_p^0$ and $\sls_p^1$ exists, for all $p\geq 6$.  On the other
hand, if $\sls_p^k$ exists, there might be multiple choices for
$\sls_p^k$.  For example, if $S={\tt abcbac}$, $\sls_2^0$ can be
either $\lsus_1^0=S[1..2]$ or $\lsus_2^0 = S[2..3]$, and $\sls_3^1$
can be any one of $\lsus_1^1=S[1..3]$, $\lsus_2^1=S[2..4]$, and
$\lsus_3^1=S[3..5]$.
Note that in Definition~\ref{def:sls}, we require $k < n$, because
finding $\sls_p^n$ is trivial as $\sls_p^n \equiv S$ for all $p$.

\begin{lemma}
\label{lem:exist}
For any $k$ and $p$: (1) $\lsus_1^k$ always exists. (2) If
$\lsus_p^k$ exists, then $\lsus_i^k$ exists, for all $i\leq p$.
 (3) If $\lsus_p^k$ does not exist, then none of $\lsus^k_{i}$ exists,
for all  $i\geq p$.
\end{lemma}

\begin{proof}
(1) $\lsus_1^k$ must exist, because at least the string $S$ can be
$\lsus_1^k$ if every proper prefix of $S$ is a $k$-mismatch repeat. 
(2) If $\lsus_p^k$ exists, say $\lsus_p^k = S[p .. q]$, $q\geq p$,
then 
$\lsus_i^k$ exists for every $i\leq p$, because at least $S[i .. q]$ is
$k$-mismatch unique.
%
(3) It is true, because otherwise we get a contradiction to the second
statement in the lemma. 
\remove{
If $\lsus_p^k$ does not exist, it means $S[p .. n]$ is a
$k$-mismatch repeat. It follows that
every suffix $S[i.. n]$ of $S[p .. n]$, $i\geq p$, is
also a $k$-mismatch repeat, i.e., none of 
$\lsus_i^k$ exists, for all $i\geq p$.}
\qed
\end{proof}

\begin{lemma}
\label{lem:lsus2}
For any $k$ and $p$, $|\lsus_p^k| \geq |\lsus_{p-1}^k| - 1$, if
$\lsus_p^k$ exists. 
\end{lemma}

\begin{proof}
\remove{
  Since $\lsus_p^k$ exists, so does $\lsus_{p-1}^k$
  (Lemma~\ref{lem:exist}).  
}
  Suppose the $k$-mismatch substring 
  $\lsus_p^k = S[p..q]$, for some $q\geq p$. Then, 
  $S[p-1..q]$ is also $k$-mismatch unique. It follows immediately that, 
  $|\lsus_{p-1}^k| \leq |S[p-1..q]| = 1+|\lsus_p^k|$.\qed
\remove{
First, the lemma is trivially correct if
  $|\lsus_{p-1}^k| \leq 2$, because $|\lsus_p^k|\geq 1$.  Next, we
  prove the lemma for the case where
  $|\lsus_{p-1}^k| \geq 3$, using contradiction.  Let's say
  $\lsus_{p-1}^k = S[p-1 .. r]$, $r>p$.  Suppose $|\lsus_{p}^k| <
  |\lsus_{p-1}^k|-1$, it means $\lsus_p^k = S[p .. q]$ for some $q$,
  $p\leq
  q<r$. Because $S[p .. q]$ is $k$-mismatch unique, so is $S[p-1
  ..
  q]$. Now we get a $k$-mismatch unique
  substring $S[p-1..q]$, which is shorter than $S[p-1..r]$. However,
  we know $S[p-1..r] = \lsus_{p-1}^k$, the shortest $k$-mismatch
  unique substring starting at position $p-1$, raising a
  contradiction. 
}
\end{proof}

\remove{
By Lemma~\ref{lem:exist}, we know that given the string $S$ and the
integer $k$, there exists a unique integer $z$, $1\leq z\leq n$, such
that $\lsus_i^k$ exists for all $i\leq z$, and none of $\lsus_j^k$
exists for all $j > z$ (if $z<n$). Further, by Lemma~\ref{lem:lsus2},
we can observe that, of those existing
$\lsus_i^k$, $i\leq z$, their
start positions strictly increase and their ending positions
monotonically increase.
}

\begin{lemma}
\label{lem:ext}
  For any $k$ and $p$, $\sus_p^k$ is either $\sls_p^k$
  or $S[i..p]$, for some $i$,
  $i+|\lsus_i^k|-1<p$. That is, $\sus_p^k$ is either the shortest $k$-mismatch
  LSUS that covers position $p$, or a right extension (through
  position $p$) of a $k$-mismatch LSUS. 
\end{lemma}

\begin{proof}
  We know $\sus_p^k$ must exist, because at least the string $S$ can
  be $\sus_p^k$. Let's say $\sus_p^k = S[i..j]$, $i\leq p \leq j$. 
  If $S[i.. j]$ is neither
  $\lsus_i^k$ nor a right extension of $\lsus_i^k$, it means
  $S[i.. j]$ is a proper prefix of $\lsus_i^k$ and thus is a
  $k$-mismatch repeat, which is a contradiction to the fact that
  $S[i.. j] = \sus_p^k$ is $k$-mismatch unique. Therefore,
  $\sus_p^k=S[i..j]$ is either $\lsus_i^k$, or a right extension of 
  $\lsus_i^k$ (clearly, $j\equiv p$ in this case).
  Further, if $\sus_p^k=S[i..j]=\lsus_i^k$, it is obvious that
  $\lsus_i^k$ must be the shortest
  $k$-mismatch LSUS covering position $p$, i.e., $\sus_p^k = \sls_p^k$.
 \qed
\end{proof}

For example, let $S={\tt dabcabc}$, then:
(1) $\sus_3^0$ can be either  $S[3..5]=\lsus_3^0$, 
or $S[1..3]$, which is a right
extension 
of $\lsus_1^0=S[1]$.
(2) $\sus_5^0 = S[4..5]=\lsus_4^0$.
(3) $\sus_6^0 = S[4..6]$, which is a right extension 
of $\lsus_4^0 = S[4..5]$.
(4) $\sus_4^1 = S[3..5] = \lsus_3^1$. 
(5) $\sus_6^1 = S[3..6]$, which is a right extension 
of  $\lsus_3^1$. 

\bigskip 

The next lemma further says that if $\sus_p^k$ is an
extension of an $k$-mismatch LSUS, $\sus_p^k$ can be quickly obtained
from $\sus_{p-1}^k$.

\begin{lemma}
\label{lem:ext2}
  For any $k$ and $p$, if $\sus_p^k = S[i..p]$ 
and  $i+|\lsus_i^k|-1 < p$, i.e., $\sus_p^k$ is a right extension
  (through position $p$) 
of
  $\lsus_i^k$, then the following must be true: (1) $p>2$; (2) the
  rightmost character of $\sus_{p-1}^k$ is $S[p-1]$; (3)
  $\sus_p^k = \sus_{p-1}^kS[p]$, the substring $\sus_{p-1}^k$
   appended by the character $S[p]$.
\end{lemma}

\begin{proof}
  If $\sus_p^k$ is a right extension (through position $p$) of a
  $k$-mismatch LSUS, it is certain
  that $p>1$, because $\sus_1^k \equiv \lsus_1^k$, which always exists
  (Lemma~\ref{lem:exist}).

  Because $\sus_p^k$ is a right extension (through position $p$) of a
  $k$-mismatch LSUS, we have
  $\sus_p^k=S[i.. p]$ for some $i < p$, and $\lsus_i^k = S[i.. j]$
  for some $j<p$.  We also know $S[i.. p-1]$ is $k$-mismatch unique, because
  the $k$-mismatch unique substring $S[i.. j]$ is a prefix of $S[i.. p-1]$.
  Note that any substring starting from a position before $i$ and
  covering position $p-1$ is longer than the $k$-mismatch unique substring
  $S[i.. p-1]$, so $\sus_{p-1}^k$ must be starting from a position
  between $i$ and $p-1$, inclusive.
  Next, we show $\sus_{p-1}^k$ actually must start at position $i$.  

  The fact that $\sus_p^k = S[i .. p]$ implies $|\lsus_t^k| \geq
  |\sus_p^k| = p-i+1$ for every $t=i+1, i+2, \ldots, p$; otherwise,
  rather than $S[i .. p]$, any one of these $\lsus_t^k$ whose size is
  smaller than $p-i+1$ would be a better choice for $\sus_p^k$. That
  means, any $k$-mismatch unique substring starting from $t=i+1, i+2,
  \ldots, p-1$ has a length at least $p-i+1$. However, $|S[i .. p-1]|
  = p-i < p-i+1$ and $S[i ..  p-1]$ is $k$-mismatch unique already and
  covers position $p-1$ as well, so $S[i .. p-1]$ is the only choice
  for $\sus_{p-1}^k$. This also means $\sus_p^k$ is indeed the
  substring $\sus_{p-1}^k$ appended by the character $S[p]$.  \qed
\end{proof}



\section{The High-Level Picture}
\label{sec:highlevel}
In this section, we present an overview of our in-place framework for
finding both the exact and approximate SUS. The framework is composed
of three stages, where all computation happens in the place of three
arrays, $S$, $A$, and $B$, each of size $n$. Arrays $A$ and $B$ are of
integers, whereas array $S$ always saves the
input string. 
The following table summarizes the roles of 
$A$ and $B$ at different stages by showing their content at the end of
each stage.

    \begin{center}
    \begin{tabular}{||c|m{7.5cm}|m{6.5cm}||}
      \hline
      \ Stages \ & \hspace*{30mm} $A[i]$ & \hspace*{30mm} $B[i]$\\
      \hline\hline
      1 & Used as temporary workspace during stage 1, but the content
      is useless for stages $2$ and $3$. & Ending position of $\lsus_i^k$, if $\lsus_i^k$ exists; otherwise, {\tt NIL}.\\
      \hline
      2  & The largest $j$, such that $\lsus_j^k$ is 
           an $\sls_i^k$, if $\sls_i^k$ exists;
      otherwise, {\tt NIL}.
      & Ending position of $\lsus_i^k$, if $\lsus_i^k$ exists; otherwise, {\tt NIL}. \\
      \hline
      3  & Start position of the rightmost $\sus_i^k$ & Ending
      position of the rightmost $\sus_i^k$\\
      \hline
   \end{tabular}
\end{center}



\noindent {\bf Stage 1} (Section~\ref{sec:lsus}). We take the array
$S$ that saves the input string as input to compute $\lsus_i^k$ for
all $i$, in the place of $A$ and $B$. At the end of the stage, each
$B[i]$ saves the ending position of $\lsus_i^k$, if $\lsus_i^k$ exists. Since each existing
$\lsus_i^k$ has its start position fixed at $i$, at the end of
stage~1, each existing $\lsus_i^k = S\bigl[i .. B[i]\bigr]$.
For those non-existing $k$-mismatch LSUSes, we assign {\tt NIL} to the
corresponding $B$ array elements. The time cost of this stage is
$O(n)$ for exact LSUS finding ($k=0$), and is $O(n^2)$ for approximate
LSUS finding, for any $k\geq 1$.

\medskip 
\noindent {\bf Stage 2} (Section~\ref{sec:sls}). Given the array $B$
(i.e., the $k$-mismatch LSUS array of $S$) from stage~1, we compute
the rightmost $\sls_i^k$, the rightmost shortest LSUS covering
position $i$, for all $i$, in the place of $A$ and $B$. At the end of
stage 2, each $A[i]$ saves the largest $j$, such that $\lsus_j^k$ is
an $\sls_i^k$, i.e., the rightmost $\sls_i^k = S\bigl[A[i]..
B[A[i]] \bigr]$, if $\sls_i^k$ exists; otherwise, we assign $A[i] =
{\tt NIL}$. Array $B$ does not change during stage~2. The time cost of
this stage is $O(n)$, for any $k\geq 0$.

\medskip 
\noindent {\bf Stage 3} (Section~\ref{sec:sus}).  Given $A$
and $B$ from stage~2,
we compute $\sus_i^k$, for all $i$, in the place of $A$ and $B$. At
the end of stage~3, each $(A[i],B[i])$ pair saves the start and ending
positions of the rightmost $\sus_i^k$, i.e., $\sus_i^k =
S\bigl[A[i].. B[i]\bigr]$. The time cost of this stage is $O(n)$,
for any $k\geq 0$.

\medskip 

Algorithms~\ref{algo:lsus-exact}, \ref{algo:lsus-approx}, 
\ref{algo:sls}, and~\ref{algo:sus} 
in the appendix give the pseudocode of the
in-place procedures that we will describe in 
Sections~\ref{sec:lsus-0}, \ref{sec:lsus-k},
\ref{sec:sls}, and~\ref{sec:sus}, respectively.

\section{Finding $k$-mismatch LSUS}
\label{sec:lsus} 
The goal of this section is that, given the input string $S$ and two
integer arrays $A$ and $B$, we want to work in the place of $A$ and
$B$, such that $B[i]$ saves the ending position of 
$\lsus_i^k$ for all existing $\lsus_i^k$;
otherwise, $B[i]$ is assigned ${\tt NIL}$.
We take different approaches in finding 
the exact LSUS ($k=0$) and approximate LSUS ($k\geq 1$).

\subsection{Finding exact LSUS ($k=0$)}
\label{sec:lsus-0}

\begin{lemma}[Lemma 7.1 in~\cite{Nong-TIS2013}]
\label{lem:sa}
Given the string $S$ of size $n$, drawn from an alphabet of
size $\sigma$, we can construct the suffix array $\sa$ of $S$ in
$O(n)$ time, using $n+\sigma$ words plus $n$ bytes, where the space of
$n$ bytes saves $S$, the space of $n$ words saves $\sa$, and the extra
space of $\sigma$ words is used as the workspace for the run of the
$\sa$ construction algorithm.
\end{lemma}

Given the input string $S$, we first use the $O(n)$-time suffix array
construction algorithm from~\cite{Nong-TIS2013} to create the $\sa$ of
$S$, where the array $A$ is used to save the $\sa$ and the array $B$
is used as the workspace. Note that $\sigma \leq n$ is always true,
because otherwise we will prune from the alphabet those characters
that do not appear in the string.
After $\sa$ (saved in $A$) is constructed, we can easily spend
another $O(n)$ time to create the rank array $\rank$ of $S$ (saved in
$B$): $\rank[\sa[i]] \leftarrow i$ (i.e., $B[A[i]] \leftarrow
i$), for all $i$.
Next, we use and work in the place of  $A$ (i.e., $\sa$) and
$B$ (i.e., $\rank$) to compute the ending position of each existing
$\lsus_i^0$ and save the result in $B[i]$, using another $O(n)$
time. 

\begin{definition}
\label{def:hl}
$$
x_i = 
\left \{
\begin{array}{ll} 
\bigl|\
\lcp^0\bigl(S[i..n],S\left[\sa\left[\rank[i]-1\right].. n\right]\bigr)\
\bigr|, & \textrm{\ \ \ if\ }\rank[i] > 1  \\
0, & \textrm{\ \ \ otherwise}
\end{array}
\right.
$$
$$
y_i = 
\left \{
\begin{array}{ll} 
\bigl|\
\lcp^0\bigl(S[i..n],S\left[\sa\left[\rank[i]+1\right].. n\right]\bigr)\
\bigr|, & \textrm{\ \ \ if\ }\rank[i] < n  \\
0, & \textrm{\ \ \ otherwise}
\end{array}
\right.
$$
That is, $x_i$ ($y_i$, resp.) is the length of the longest
common prefix of $S[i..n]$
and its lexicographically preceding (succeeding, resp.)
suffix, if the preceding (succeeding, resp.) suffix exists. 
\end{definition}

\begin{fact}
\label{fact:lsus-length}
For every string position $i$,  $1\leq i \leq n$: 
$$
\lsus_i^0 = \left\{
  \begin{array}{ll}
    S\left[i.. i+\max\{x_i,y_i\}\right], & \textrm{\ \ if } i+\max\{x_i, y_i\} \leq n \\
    \textrm{not existing}, & \textrm{\ \ otherwise.}
  \end{array}
\right.
$$
\end{fact}

First, observe that in the sequence of $x_i$'s, if $x_i > 0$, then $x_{i+1}
\geq x_i-1$ must be true, because at least $S[\sa[\rank[i]-1]+1 ..
n]$ can be the lexicographically preceding suffix of $S[i+1.. n]$,
and they share the leading $x_i-1$ characters.  That means, when we compute
$x_{i+1}$, we can skip over the comparisons of the first $x_i-1$ pair
of characters between $S[i+1.. n]$ and its lexicographically
preceding suffix. It follows that, given the $\sa$ and $\rank$ of $S$ and using the
above observation, we can compute the sequence of $x_i$'s in $O(n)$
time. 
%
Using the similar observation, we can compute the sequence of $y_i$'s in
$O(n)$ time, provided that $S$ and its $\sa$ and $\rank$ are given.

Second, since we can compute the sequences of $x_i$'s and $y_i$'s in
parallel (i.e., compute the sequence of ($x_i,y_i$) pairs), we can use
Fact~\ref{fact:lsus-length} to compute the sequence of $\lsus_i^0$ in
$O(n)$ time.  Further, since $\rank[i]$ is used only for retrieving
the lexicographically preceding and succeeding suffixes of $S[i..n]$
when we compute the pair ($x_i, y_i$), we can save each computed
$\lsus_i^0$ (indeed, $i+\max\{x_i,y_i\}$, the ending position of
$\lsus_i^0$) in the place of $\rank[i]$ (i.e., $B[i]$).  In the case
$i + \max\{x_i,y_i\} > n$, meaning $\lsus_i^0$ does not exist,we will
assign {\tt NIL} to $\rank[j]$ (i.e., $B[j]$) for all $j\geq i$
(Lemma~\ref{lem:exist}). The overall time cost for computing the
sequence of $\lsus_i^0$ is thus $O(n)$, yielding the following lemma.

\begin{lemma}
\label{lem:lsus-exact}
  Given the character array $S$ of size $n$ that saves the input string, 
  and the integer arrays $A$ and $B$, each of size $n$, we can work
  in the place of $S$, $A$, and $B$, using $O(n)$ time, such that at
  the end of the computation, $S$ does not change, $B[i]$ saves the
  ending of position of $\lsus_i^0$, if $\lsus_i^0$ exists (otherwise,
  $B[i]={\tt NIL}$).
\end{lemma}

\subsection{Finding approximate LSUS ($k \geq 1$)}
\label{sec:lsus-k}

\begin{definition}
\label{def:llr}
For a particular string position $p$ in $S$ 
and an integer $k$, $0\leq k\leq n-1$, 
\emph{the $k$-mismatch left-bounded longest repeat (LLR)}
starting at position $p$, denoted as $\llr_p^k$, is a $k$-mismatch
repeat $S[p..j]$, 
such that either $j=n$ or $S[p .. j+1]$ is
$k$-mismatch unique.   
\end{definition}

\begin{fact}
\label{fact:llr}
(1) If $|\llr_p^k|< n - p + 1$, i.e., the ending position of $\llr_p^k$ is
less than $n$, then $\lsus_p^k = S\bigl[p ..
p+|\llr_p^k|\bigr]$, the substring of
$\llr_p^k$ appended by the character following $\llr_p^k$.
(2) Otherwise, $\lsus_p^k$ does not exist. 
\end{fact}

Our high-level strategy for finding $\lsus_i^k$ for all $i$ is as
follows.  We first find $\llr_i^k$ for all $i$. Then we use
Fact~\ref{fact:llr} to find each $\lsus_i^k$ from $\llr_i^k$: If
$\llr_i^k$ does not end on position $n$, we will extend it 
for one more character on its right side and make the extension to be
$\lsus_i^k$; otherwise, $\llr_i^k$ does not exist. Next, we 
explain how to find $\llr_i^k$, for all $i$. 

Clearly, $|\llr_i^k| = \max\{|\lcp^k(S_i, S_j)|, j\neq i\}$, for all
i. The way we calculate $|\llr_i^k|$ for all $i$ is simply to let
every pair of two distinct suffixes to be compared with each other. In
order to do so, we work over $n-1$ phases, named as $\mathcal{P}_1$
through $\mathcal{P}_{n-1}$. On a particular phase
$\mathcal{P}_\delta$, we compare suffixes $S_i$ and $S_{i-\delta}$ for
all $i=n,n-1,\ldots,\delta+1$. Obviously, over these $n-1$ phases,
every pair of distinct suffixes have been compared with each other
exactly once. 
Over these $n-1$ phases, we simply record in $B[i]$, which is
initialized to be $0$, the length of 
the longest $k$-mismatch LCP that each suffix $S_i$
has seen when compared with any other suffixes. 
Next, we explain the details of a particular phase
$\mathcal{P}_\delta$.

On a particular phase $\mathcal{P}_\delta$, $1\leq \delta\leq n-1$, we
compare suffixes $S_i$ and $S_{i-\delta}$ for all
$i=n,n-1,\ldots,\delta+1$. When we compare $S_i$ and $S_{i-\delta}$,
we save in $A[1.. k+1]$, which is initialized to be empty at the
beginning of each phase, the leftmost mismatched $k+1$ positions in
$S_i$. We will see later how to update $A[1.. k+1]$ efficiently
over the progress of a particular phase and use it to update the $B$
array. 

We treat $A[1.. k+1]$ as a circular array,
i.e., $i-1 = k+1$ when $i=1$, and $i+1=1$ when $i=k+1$. Let ${\tt
  size}$, which is initialized to be $0$ at the beginning of each
phase, denote the number of mismatched positions being saved in
$A[1.. k+1]$ so far in $\mathcal{P}_\delta$. We can describe the
work of phase $\mathcal{P}_\delta$, inductively, as follows.

\begin{enumerate}
\item We compare $S_n$ and $S_{n-\delta}$ by only comparing 
$S[n]$ and $S[n-\delta]$, since $S_n = S[n]$.
\begin{enumerate}
\item If $S[n] \neq S[n-\delta]$:
Save $n$ in any position in $A[1.. k+1]$;
${\tt size} \leftarrow 1$.

\item 
$B[n]\leftarrow \max\{B[n], 1\}$;
$B[n-\delta]\leftarrow \max\{B[n-\delta], 1\}$.

\end{enumerate} 

\item 
  Suppose we have finished the comparison between the suffixes
  $S_{i+1}$ and $S_{i+1-\delta}$, for some $i$, $\delta+1 \leq i \leq n-1$.
  The leftmost $k+1$ mismatched positions (if existing) between them
  have been saved in the circular array $A[1.. k+1]$. Let $A[{\tt cursor}]$
  be the element that 
  is saving the first mismatched position (if existing) between the two
  suffixes. 

\item  Next, we compare the suffixes $S_i$ and $S_{i-\delta}$ by
  only comparing $S[i]$ and $S[i-\delta]$, since $S_{i+1}$ and
  $S_{i+1-\delta}$ have been compared. Remind that ${\tt cursor}-1$
  below is in its cyclic manner.

  \begin{enumerate}
  \item If $S[i]\neq S[i-\delta]$: ${\tt cursor}\leftarrow {\tt
      cursor}-1$;  Save $i$
    in $A[{\tt cursor}]$ and overwrite the old content in $A[{\tt
      cursor}]$ if
    there is; ${\tt size} \leftarrow \min\{{\tt size+1}, k+1\}$.

  \item If ${\tt size} < k+1$: 
    $B[i]\leftarrow \max\{B[i], n-i+1\}$;
    $B[i-\delta]\leftarrow \max\{B[i-\delta], n-i+1\}$.
  \item Else:
    $B[i]\leftarrow \max\{B[i], A[{\tt cursor}-1]-i\}$;
    $B[i-\delta]\leftarrow \max\{B[i-\delta], A[{\tt cursor}-1]-i\}$.
    Note that $A[{\tt cursor}-1]$ is saving the $(k+1)$th
    mismatched position between $S_i$ and $S_{i-\delta}$.
\end{enumerate}

\end{enumerate}

After the computation of all $\llr_i^k$ is finished, using the above 
$n-1$ phases, each $B[i]$ is saving $|\llr_i^k|$.
Next,
we can use Fact~\ref{fact:llr} to convert each $\llr_i^k$ to
$\lsus_i^k$ by simply checking each $B[i]$: If $i+B[i]-1<n$, i.e.,
$\llr_i^k$ does not end on position $n$, then we assigne $B[i] =
i+B[i]$, the ending position of $\lsus_i^k$; otherwise, we assign
$B[i] = {\tt NIL}$, meaning $\lsus_i^k$ does not exist.

\medskip 

The computation 
of all $\llr_i^k$ takes
 $n-1$ phases and each phase clearly has no more
than $n$ comparisons, giving a total of $O(n^2)$ time cost. 
The procedure of converting each $\llr_i^k$ to $\lsus_i^k$ spends
another $O(n)$ time. Altogether, we get an $O(n^2)$-time in-place
procedure for finding approximate LSUS, for any $k\geq 1$.

\begin{lemma}
\label{lem:lsus-k}
  Given the character array $S$ of  size $n$ that saves the input string, 
  the integer arrays $A$ and $B$, each of size $n$, and the value
  of integer $k\geq 1$, we can work
  in the place of $S$, $A$, and $B$, using $O(n^2)$ time, such that at
  the end of the computation, $S$ does not change, $B[i]$ saves the
  ending of position of $\lsus_i^k$,
i.e., $\lsus_i^k = S[i.. B[i]]$,
 if $\lsus_i^k$ exists; otherwise,
  $B[i]={\tt NIL}$.
\end{lemma}

\section{Finding $k$-mismatch SLS}
\label{sec:sls}
Now we are given the array $B$, where each $B[i]$ saves the ending
position of $\lsus_i^k$ if $\lsus_i^k$ exists and {\tt NIL}
otherwise. In this section, we want to work in the place of $A$ and
$B$, such that in the end of computation: $A[i]$ saves $j$, such that
$\lsus_j^k$ is the rightmost $\sls_i^k$, if such $j$ exists;
otherwise, $A[i]={\tt NIL}$.
That means, in the end of this section, the rightmost $\sls_i^k =
S\bigl[A[i].. B[A[i]]\bigr]$, if $\sls_i^k$ exists; otherwise,
$A[i] = B[i] = {\tt NIL}$.

Recall that some $k$-mismatch LSUS may not exist and some positions
may not be covered by any $k$-mismatch LSUS (see the examples after
Definition~\ref{def:lsus}). Further, due to Lemmas~\ref{lem:exist}
and~\ref{lem:lsus2}, we know such positions that are not covered by any
$k$-mismatch LSUS must comprise a continuous chunk on the right end of
string $S$.

\begin{definition}
\label{def:rz}
Let $\lsus_r^k$, $1\leq r\leq n$, be the rightmost existing
$k$-mismatch LSUS of the input string $S$.
Let $z$, $1\leq z\leq n$, be the rightmost string position that is
covered by any $k$-mismatch LSUS  of the string $S$.
\end{definition}

Again, due to Lemmas~\ref{lem:exist} and~\ref{lem:lsus2}, it is trivial
to find the values of $r$ and $z$ in $O(n)$ time: scan array $B$ (i.e.
LSUS array) from right to left, and stop when seeing the first
non-NIL $B$ array element, which is exactly $B[r]$, then $z = B[r]$.
If $z<n$, we can then simply set $A[i]={\tt NIL}$
for all $i>z$. Recall that $B[i]={\tt NIL}$
already for all $i>r$ from stage~1. In the rest of this section, we 
only need to work with the two subarrays $A[1..z]$ and $B[1..z]$, 
wanting to make $A[i]$ to be the start position of the rightmost
$\sls_i^k$, for all $i\leq z$.

Let $B[1..z]$ and an integer $r$, $1\leq r\leq z$, be the input, where
(1)~$B[1..r]$ is of monotonically nondecreasing integers
(Lemma~\ref{lem:lsus2}), with $i \leq B[i]$, (2)~$B[r+1.. z]$ are
all {\tt NIL}, if $r<z$, and (3)~$B[r]=z$.

We can use each $B[i]$, $i\leq r$, as a compact representation of the
interval $I_i=(i, B[i])$. Let ${\cal I} = \{\; I_i \mid i \in [1..r] \;
\}$, and $\ell_i = |B[i] - i + 1|$ be the length of $I_i$. Let
$A[1..z]$ be an output array such that $A[j]=i$, where $I_i$ is the
rightmost shortest interval in ${\cal I}$ that covers $j$.

\medskip

%
To illustrate the ideas and
concepts that we will present in the rest of this section, 
let us use the following as a running example, 
where $r=9$, $z=15$, and $n=17$
(we add
$(0,B[0])=(0,0)$ as a sentinel).

{\small
\begin{center}
$
\begin{array}{||c||c|c|c|c|c|c|c|c|c|c|c|c|c|c|c|c|c|c||}
\hline 
i & 0   &  1 &  2 &  3 &  4 &  5 &  6 &  7 &  8 &  9 & 10 & 11 & 12 &
13 & 14 & 15 & 16 & 17 \\\hline
B[i] & 0 &  3 &  4 &  7 &  8 & 10 & 10 & 10 & 11 & 15 & - & - & - & -
& - & - & - & - \\\hline
\ell_i & 0 &  3 &  3 &  5 &  5 &  6 &  5 &  4 &  4 &  7 & - & - & - & -
& - & - & - & - \\ \hline \hline
\pred[i] & - & -& -&  2 &  2 &  4 &  2 &  2 &  2 &  8 & -& -& -& -& -&
-& -& -\\\hline
t_i  & - &  1 &  2 &  5 &  5 &  9 &  6 &  7 &  8 & 12 & -& -& -& -& -&
-& -& -\\\hline
\max\ t_i^{-1} & - &  1 &  2 & -& -&  4 &  6 &  7 &  8 &  5 & -& -&  9 & -&
-& -& -& -\\\hline\hline
A[i]  & - &  1 &  2 &  2 &  2 &  4 &  6 &  7 &  8 &  8 &  8 &  8 &  9
&  9 &  9 &  9 & -& -\\\hline
\end{array}
$
\end{center}
}

\remove{
\begin{center}
$
\begin{array}{||c||c|c|c|c|c|c|c|c|c|c||}
\hline 
i           & 0 & 1 & 2 & 3 & 4 & 5 &  6 &  7  & 8   &  9 \\ \hline
B[i]      & 0 & 3 & 4 & 7 & 8 & 10 & 10 & 10 & 11 & 15 \\ \hline
\ell_i     & 0 & 3 & 3 & 5 & 5 & 6 & 5  & 4 & 4 & 7 \\ \hline \hline
\pred[i]   & - & - & - & 2 & 2 & 4 & 2 & 2 & 2 & 8 \\ \hline
t_i         & - & 1 & 2 & 5 & 5 & 9 & 6 & 7 & 8 & 12 \\ \hline
t_i^{-1} & - & 1 & 2 & - & - & 3, 4 & 6 & 7 & 8 & 5 \\ \hline 
\max\ t_i^{-1} & - & 1 & 2 & - & - & 4 & 6 & 7 & 8 & 5 \\ \hline \hline 
A[i]       & - & 1 & 2 & 2 & 2 & 4 & 6 & 7 & 8 & 8 \\ \hline 
\end{array}
$
\end{center}
}

\begin{definition}
For an interval $I_i$, we define the \emph{effective covering region} 
with respect to the previous intervals 
${\cal I}_{< i} = \{\; I_k \mid k < i \; \}$ to be $\bigl[t_i, B[i]\bigr]$ where 
\[
t_i = \max\;  \Big\{\; i,\;\; \max\; \{ B[k] + 1 \mid I_k 
\mbox{\ is shorter than\ } I_i, k < i \; \}\; \Big\}. 
\]
We call $t_i$ the starting point of the effective covering region of $I_i$.
\end{definition}
The effective covering region of $I_i$ is exactly those regions that
would set $I_i$ as the answer, provided that all the intervals $I_{<
  i}$ before $I_i$ are present, and all the intervals $I_{> i} = \{\;
I_k \mid k > i\; \}$ are absent.

\medskip

\noindent
We next define $t_i^{-1}$ as a list\footnote{In actual run,
  $t_i^{-1}$ saves the largest number in that list, as we will see
  more clearly later.}, such that $j \in
t_i^{-1}$ if and only if $t_j = i$. Observe that since $t_i \geq i$ by
definition, any value $j$ in $t_i^{-1}$ must have $j \leq i$, and the
effective region of $I_j$ must cover $i$.

\begin{lemma}
\label{lem:b}
For $i = 1,2,\ldots, z$:
\[A[i] = \max\; \bigcup_{k=1}^i t_k^{-1} = \max\; 
\{\; A[i-1],\; \max\; t_i^{-1}\; \}.\]
\end{lemma}
\begin{proof}
  Let $j = \max\; \bigcup_{k=1}^i t_k^{-1}.$ This means that for the
  effective region of any $I_h$, with $h > j$, none of them covers
  $i$.  Next, observe that $I_j$ must cover $i$; otherwise, for all
  the intervals $I_h$ with $h < j$, we have $B[h] \leq B[j] < i$, so
  that none of them can cover $i$, and thus a contradiction occurs.
Finally, we show that for those $h < j$, $I_h$ can be
    pruned by $I_j$, thus implying that $A[i] = j$ is a correct
    answer.

\medskip

\noindent
Consider all those $h$ with $h < j$:  
\begin{enumerate}
\item If $I_h$ is longer than $I_j$, $I_h$ can be pruned away
  directly.

\item Else, if $I_h$ and $I_j$ have equal length, $I_h$ can be pruned
  away also, regardless of its coverage on $i$,
since we pick the rightmost shortest interval that covers
  $i$.

\item Else, $h$ must appear in $\bigcup_{k=1}^i t_k^{-1}$. By the
  definition of $t_j$, we have $B[h] < t_j \leq i$; thus, $I_h$ does
  not cover $i$, and can be pruned away.
\end{enumerate}
\remove{
\begin{enumerate}
\item If $I_h$ is no shorter than $I_j$, $I_h$ can be pruned away
  directly.

\item Else, if $h$ appears in $\bigcup_{k=1}^i t_k^{-1}$, by the
  definition of $t_j$, we have $B[h] < t_j \leq i$; thus. $I_h$ does
  not cover $i$, and can be pruned away.

\item Else, $h$ does not appear in $\bigcup_{k=1}^i t_k^{-1}$.  Then,
  we have $t_h > i$, so that the effective region of $I_h$ does not
  cover $i$; so $I_h$ can be pruned away.
\end{enumerate}
}
Thus, the first equality in the lemma follows, while the second
equality in the lemma is trivial once we have the first equality. 
\qed 
\end{proof}

\begin{lemma}
Suppose that all $t_i$, $1\leq i\leq r$, can be 
generated incrementally in $O(n)$ time.  
Then, we can obtain all $\max\ t_i^{-1}$, $1\leq i\leq z$, 
in $O(n)$ time.
\end{lemma}

\begin{proof}
  We examine each $t_i$, $i=1,2,\ldots,r$, and write $i$
  at entry $t_i = j$ of
  the $t^{-1}$ array; if such an entry contains a value $i'$ already,
  we simply overwrite $i'$ with the latter $i$. \qed
\end{proof}

\noindent
Indeed, we may scan $t_i$ from right to left, i.e., $i=r,r-1,\ldots,
1$, and update $\max\ t_i^{-1}$ as we proceed. Firstly, if $t_i > i$,
we set $t_i^{-1} = \mbox{undefined}$. Else, let $j = t_i$ (whose value
is at least $i$), and we check if $t^{-1}_{j}$ is defined: If not,
simply set $t^{-1}_j = i$; otherwise, no update is needed.

\smallskip

\noindent
The advantage of the `right-to-left' approach is that we can construct
$t_i^{-1}$ in-place, by re-using the memory space of $t_i$.  To see
why it is so, by the time we need to update a certain entry $j = t_i$
at step $i$, the information $t_j$ has been used (and will never be
used), so that we can safely overwrite the original entry, storing
$t_j$, to store $t_j^{-1}$ instead.  This gives the following
corollary.

\begin{corollary}
\label{cor:t}
  Suppose that all $t_i$'s are generated, and are stored in a certain array
  $A[1..z]$. Then, we can obtain $\max\ t_i^{-1}$ for all $i$'s,
  in-place,
  by storing the results in the same array $A[1..z]$; the time cost is
  $O(n)$.
\end{corollary}

\noindent
Our goal is to make our algorithm in-place. Suppose that we can have
in-place incremental generation of $t_i$. Then, by the above lemma, we
may store $\max\ t_i^{-1}$ temporarily at $A[i]$; afterwards, by the
second equality of Lemma~\ref{lem:b}, we can compute the correct
output $A$ by a simple scan of $A$ from left to right.

Thus, to make the whole process in-place, it remains to show how $t_i$
can be computed in $O(n)$ time, in-place.  For this, we define
$\pred[i]$ to be the largest $j$ (if it exists) such that $j < i$ and
length of $I_j$ is shorter than $I_i$.  It is easy to check that if
$\pred[i] = j$ is defined, then $t_i = \max\; \{\; B[j]+1, i \; \}$
(and $t_i = i$ otherwise).\footnote{%
  For each $j' < j$, if $I_{j'}$ covers $i$, $I_j$ would also cover
  $i$; in such a case, $B[j] + 1 \geq B[j'] + 1$.  For each $j' \in
  [\pred[i], i-1]$, $I_{j'}$ is longer than $I_i$.}  Moreover,
$\pred[i]$ for all $i$'s can be computed incrementally, with a way
analogous to the construction of the failure function in KMP
algorithm: we check $\pred[i-1], \pred[\pred[i-1]],
\pred[\pred[\pred[i-1]],$ and so on, until we obtain $j$ in the
process such that $I_j$ is shorter than $I_i$, and set $\pred[i] :=
j$.\footnote{Intuitively, $\pred$ defines the shortcuts so that we can
  skip some intervals in $I_{<i}$ to compute $t_i$.} If such $j$ does
not exist, we set $\mathtt{pred}[i]=\mathtt{NIL}$. The running time
is bounded by $O(n)$.

\bigskip

\noindent
This gives the following $O(n)$-time in-place algorithm (where $B$ is read-only):
\begin{enumerate}
  \setlength{\itemsep}{0pt} \item Compute $\pred[i]$,
  $i=1,2,\ldots,r$, and store this in $A[i]$.
  Note that this step requires the length information of the intervals
  of $I_i$, which can be obtained in $O(1)$ time, on the fly, from
  $B[i]$ .
\item Scan $A[1..r]$ (i.e., $\pred$) incrementally,
  and obtain $t_i$ from
  the above discussion.  The value of $t_i$ is stored in $A[i]$.  Note
  that this step requires the access to the original $B$.
\item Scan $A[1..r]$ (i.e., $t_i$) from right to left, and obtain
  $\max\ t_i^{-1}$ decrementally (stored in $A[i]$) by
  Corollary~\ref{cor:t}. 
\item Scan $A[1..z]$ (i.e., $\max\ t_i^{-1}$)
  incrementally ($i=1,2,\ldots,z$), and obtain the
  desired $A[i]$ by the second equality in Lemma~\ref{lem:b}.
\end{enumerate}

\begin{lemma}
\label{lem:sls}

Given the integer array $A$ and $B$, each of size $n$, 
where each $B[i]$ saves the ending
position of $\lsus_i^k$, if $\lsus_i^k$ exists and {\tt NIL}
otherwise, we can work in the place of
array $A$ and $B$, using $O(n)$ time, such that, 
in the end of computation,
array $B$ does not change, and $A[i]$ saves $j$, where
$\lsus_j^k$ is the rightmost $\sls_i^k$, 
if such $j$ exists; otherwise, $A[i]={\tt NIL}$.
That is, 
$\sls_i^k = S\bigl[A[i].. B[A[i]]\bigr]$, if $\sls_i^k$ exists;
otherwise, $A[i] = B[i] = {\tt NIL}$.
\end{lemma}

\section{Finding $k$-mismatch SUS}
\label{sec:sus}
Now we have array $A$, where $A[i]=j$, such that $\lsus_j^k$ is the
rightmost $\sls_i^k$, if position $i$ is covered by any
$k$-mismatch LSUS; otherwise, $A[i] = {\tt NIL}$. Note that $A[i] = j$
is recording the start position of the rightmost $\sls_i^k$ already,
because $\lsus_j^k$ starts on position $j$.  We also have array $B$,
where $B[i] = i+|\lsus_i^k|-1$, the ending position of $\lsus_i^k$, if
$\lsus_i^k$ exists; otherwise, $B[i]={\tt NIL}$.

\remove{
Recall that it is possible that some positions may not be covered by
any $k$-mismatch LSUS (see the example after
Definition~\ref{def:lsus}). Due to Lemma~\ref{lem:exist}
and~\ref{lem:lsus2}, we know such positions must comprise a continuous
chunk on the right end of string $S$. Let $z$, $1\leq z\leq n$, denote
the rightmost position that is covered by at least one $k$-mismatch
LSUS, i.e., $A[z]$ is the rightmost non-NIL element in array
$A$. Again, due to Lemma~\ref{lem:exist} and~\ref{lem:lsus2}, it is
trivial to find the value of $z$ in $O(n)$ time: scan
array $B$ from right toward left, and stop when meeting the first
non-NIL $B$ array element, say $B[\delta]$, then $z = \delta +
B[\delta] -1$.
}

\paragraph{Step I.} We want to transform $A$ and $B$, such that each
$(A[i],B[i])$ pair saves the start and ending positions of $\sls_i^k$,
if $\sls_i^k$ exists; otherwise, we set $(A[i],B[i]) = ({\tt
  NIL,NIL})$.  Since each $A[i]$ is already recording the start
position of $\sls_i^k$ already, as we have explained at the beginning
of this section, we only need to make changes to array $B$.  We first
set $B[i]={\tt NIL}$ for all $i>z$ (Definition~\ref{def:rz}). Then, we
scan array $B$ from right to left, starting from position $z$ through
$1$, and set each $B[i] = B[A[i]]$, the ending position of the
rightmost $\sls_i^k$. Because the leftmost position that any existing
$\lsus_i^k$ can cover is position $i$, we know $A[i]\leq i$ and we no
longer need $B[i]$ (i.e., the information of $\lsus_i^k$) after
$\sls_i$ is computed. Therefore, it is safe to record $\sls_i^k$ by
overwriting $B[i]$ by $B[A[i]]$ (i.e., the ending position of
$\sls_i^k$), in this right-to-left scan.

\paragraph{Step II.} 
We use arrays $A$ and $B$ to calculate $\sus_i^k$ for each $i$
and save the result in the place of $A$ and $B$, i.e., each
$(A[i],B[i])$ pair saves the start and ending position of $\sus_i^k$.
Because of Lemma~\ref{lem:ext} and \ref{lem:ext2}, we can use arrays
$A$ and $B$ to compute each $\sus_i^k$ inductively, as follows: 
\begin{enumerate}
\item 
$\sus_1^k = \lsus_1^k = \sls_1^k = S\bigl[A[1].. B[1]\bigr]$. 
\item 
For $i=2,3,\ldots,n$, we compute $\sus_i^k$: 
\begin{enumerate}
\item 
If $(A[i],B[i]) = ({\tt NIL}, {\tt NIL})$, meaning $\sls_i^k$
does not exist, we set $\sus_i^k$ to be $\sus_{i-1}^k$ appended by the
character $S[i]$, i.e., $\sus_i^k = S\bigl[A[i-1]..
B[i-1]+1\bigr]$, and save 
$\sus_i^k$
by setting $(A[i],B[i]) = (A[i-1],B[i-1]+1)$; 
\item Else, if $\sus_{i-1}^k$ ends at position $i-1$ and 
$\sus_{i-1}^kS[i]=S\bigl[A[i-1].. B[i-1]+1\bigr]$ is shorter 
than $\sls_i^k = S\bigl[A[i].. B[i]\bigr]$, 
we set $(A[i],B[i]) = (A[i-1],B[i-1]+1)$; 
\item Else, $\sus_i^k = \sls_i^k$ and thus we leave $A[i]$
  and $B[i]$ unchanged.
\end{enumerate}
\end{enumerate}

\begin{lemma}
\label{lem:sus}  
Given arrays $A$ and $B$:
\begin{itemize}
\item 
$A[i]=j$, such that $\lsus_j^k$ is the
rightmost $\sls_i^k$, if $\sls_i^k$ exists; 
otherwise, $A[i] = {\tt NIL}$; 
\item $B[i] = i+|\lsus_i^k|-1$, the ending position of $\lsus_i^k$, if
$\lsus_i^k$ exists; otherwise, $B[i]={\tt NIL}$. 
\end{itemize}
we can work in the place of $A$ and $B$, using $O(n)$ time, 
such that, in the end of
computation, each $(A[i],B[i])$ saves the start and ending
positions of $\sus_i^k$, i.e., 
$\sus_i^k = S\bigl[A[i].. B[i]\bigr]$, $i=1,2,\ldots,n$.  
\end{lemma}

\noindent
By concatenating the claims in Lemmas~\ref{lem:lsus-exact},
\ref{lem:lsus-k}, \ref{lem:sls}, and \ref{lem:sus}, we get the
final result.

\begin{theorem} 
\label{thm:sus} 
Given the array $S$ of size $n$ that
  saves the input string, two empty integer arrays $A$ and $B$, each
  of size $n$, and the value of integer $k\geq 0$, we can work in the
  place of arrays $S$, $A$, and $B$, using a total of $O(n)$ time for
  $k=0$ and $O(n^2)$ time for any $k\geq 1$, such that in the end of
  computation, $S$ does not change, each $(A[i], B[i])$ pair
  represents the start and ending positions of the rightmost
  $\sus_i^k$, i.e., $\sus_i^k = S\bigl[A[i].. B[i]\bigr]$. 
\end{theorem}

\remove{

\section{Extension: finding all SUSes for every position}
The current solution for both exact and mismatched SUS finding only
find one candidate for each $\sus_i$. Can we find all candidates for
each $\sus_i$ ? Of course, doing so will need more than $2n$ words.

\section{Extension: finding  SUS covering chunk of position with k-mm}

The exact SUS version covering chunk of positions 
has been studied by HU et al in SPIRE 2014. The
solution to the
mm version can be the following. 

Compute all the existing LSUS's and RSUS's (right-bounded shortest
unique substring). We have no more than $2n$ intervals here. Delete
all duplicated intervals if there are. Given any chunk of positions,
its SUS is the shortest interval (or the extension of some interval)
that cover the chunk of positions. The exact-SUS problem can eventually be reduced to
range-minimum-query problem. I somehow believe this is also possible
for the mm version. It shouldn't take much time to think through it if we
want to add this extension, which of course won't be in the $2n$ words $+$ $n$ Bytes
framework though. 

}

\section{Conclusion}
In this paper, we revisited the exact SUS finding problem, and
proposed its approximate version where mismatches are allowed, and
thus significantly extended the usage of SUS finding in subfields such
as computational biology. We designed a generic in-place algorithmic
framework that uses the minimum $2n$ words plus $n$ bytes space and
can fit to find both exact and approximate $k$-mismatch SUS, with
$O(n)$ and $O(n^2)$ time complexities, respectively, regardless of the
value of any $k\geq 1$. 
An urgent future work will be researching for a faster (and still
practical) in-place algorithm for finding approximate LSUS to replace
the current algorithm discussed in Section~\ref{sec:lsus-k}.  Such new
algorithm will lead to an overall faster in-place solution for
approximate SUS finding.

\small
\bibliographystyle{splncs03}
\bibliography{bibjsv,repeat,pm}

\newpage

\appendix

\section*{Appendix}


\begin{algorithm}[h!]
 \caption{Finding exact LSUS}
\label{algo:lsus-exact}
\KwIn{String $S$ and integer arrays $A$ and $B$, each of size $n$.} 
\KwOut{$S$ does not change. $B[i]$ = ending position of
  $\lsus_i^0$, if $\lsus_i^0$ exists; otherwise, $B[i]={\tt NIL}$.}

\bigskip 

Create the $\sa$ of $S$ using the suffix array construction algorithm
from~\cite{Nong-TIS2013}, where array $A$ is used to save the
resulting $\sa$ and $B$ is used as the workspace for the run of the
algorithm. 

\medskip 

\tcc{Create the $\rank$ of $S$ and save the result in the array $B$.}
\lFor{$i=1\ldots n$}{$\rank[\sa[i]] \leftarrow i$}
\tcp*{i.e., $\sa = A$, $\rank=B$, and $B[A[i]]\leftarrow i$.}

\medskip

\tcc{From here on,  $A$ and $\sa$ are the same physical array. 
$B$, $\rank$, and $\lsus$ are the same physical array.}

\medskip 

$x \leftarrow 0$;
$y \leftarrow 0$\;

\For{$i = 1, 2, \ldots, n$}{
  \If{$\rank[i] > 1$}{
    $j\leftarrow \sa[\rank[i]-1]$\;
    \tcc{Calculate the length of the $0$-mismatch LCP
        between  $\seq[i.. n]$ and its lexicographically preceding suffix.}
    \lWhile{$\seq[i+x] = \seq[j+x]$}
    {$x \leftarrow x + 1$}
  }
  \lElse{
    $x\leftarrow 0$\;
  }
  \If{$\rank[i] < n$}{
    $j\leftarrow \sa[\rank[i]+1]$\;
    \tcc{Calculate the length of the $0$-mismatch LCP
        between  $\seq[i.. n]$ and its lexicographically succeeding suffix.}
   \lWhile{$\seq[i+y] = \seq[j+y]$}{$y \leftarrow y + 1$\;}
  }
  \lElse{
    $y\leftarrow 0$\;
  }
  \lIf
   {$i+\max\{x,y\}\leq n$}{
    $\lsus[i] \leftarrow i + \max\{x, y\} $\tcp*{ending position of $|\lsus_i|$}
  }
  \Else(\tcp*[f]{$\lsus_i$ does not exist. Early stop.}){
    \lFor{$j=i\ldots n$}{
      $\lsus[j] \leftarrow \nil$\;
    }
    Break\;
  }

  \lIf{$x > 0$}{$x \leftarrow x - 1$\;}
  \lIf{$y > 0$}{$y \leftarrow y - 1$\;}
}
\end{algorithm}


\begin{algorithm}[h!]
 \caption{Finding approximate LSUS}
\label{algo:lsus-approx}
\KwIn{String $S$ and integer arrays $A$ and $B$, each of size $n$, the
  value of $k\geq 1$.} 
\KwOut{$S$ does not change. $B[i]$ = ending position of
  $\lsus_i^k$, if $\lsus_i^0$ exists; otherwise, $B[i]={\tt NIL}$.}

\bigskip 

\lFor{$i=1\ldots n$}{
  $B[i] \leftarrow 0$ \tcp*{Initialization}
}

\medskip

\tcc{We use $A[1\dots k+1]$ as a circular array to save the $k+1$ most recently
found mismatched positions.}

\medskip 

$capacity \leftarrow k+1$\tcp*{The capacity of the circular array that 
  records at most $k+1$ mismatched positions.}

\medskip

$cursor \leftarrow 1$\tcp*{The index of the circular array position that is
  saving the most recently founded mismatched position. It can be
  initialized to be any value from $\{1,2,\ldots, capacity\}$.}

\medskip

\For(\tcp*[f]{$n-1$ phases}){$\delta = 1 \ldots n-1$}{
  $size\leftarrow 0$ \tcp*{The number of recorded mismatched positions
    in the circular array in the current phase.}
    \medskip
  \For{$i=n$ down to $\delta+1$}{
    \medskip
    \tcc{Comparing suffixes $S_i$ and $S_{i-\delta}$ by comparing
      their leading characters, as their remaining characters 
      have been
      compared in previous steps of this phase.}
    \medskip
   \If{$S[i] \neq S[i-\delta]$}{
      $cursor \leftarrow \bigl((cursor - 2 + capacity) \mod capacity\bigr) +
      1$\tcp*{We use 1-based indexing.}
      $A[cursor] \leftarrow i$\;
      $size \leftarrow \min(size + 1, capacity)$; 
    }
    \If{$size < capacity$}{
      $B[i] \leftarrow \max(B[i], n -i+1)$   \tcp*{$=size - i + 1$}
      $B[i-\delta] \leftarrow \max(B[i-\delta], n -i+1)$\;
    }
    \Else{
      $B[i] \leftarrow \max\Bigl(B[i], A\bigl[\bigl((cursor-1+k)\mod capacity\bigr)+1\bigr] -
      i\Bigr)$\tcp*{We use 1-based indexing.}
      $B[i-\delta] \leftarrow \max\Bigl(B[i-\delta], A\bigl[\bigl((cursor-1+k)\mod capacity\bigr)+1\bigr] -
      i\Bigr)$;
    }
  }
}

\For{$i = 1 \ldots n$}{
  \lIf{$B[i] = size - i +1$}{
    $B[i] \leftarrow \nil$\tcp*{$\lsus_i^k$ does not exist.}
  }
  \lElse{
    $B[i] \leftarrow i + B[i]$\tcp*{The ending position of $\lsus_i^k$.}
  }
}
\end{algorithm}


\begin{algorithm}[h!]
 \caption{Finding SLS (exact or approximate)}
\label{algo:sls}
\KwIn{Integer arrays $A$ and $B$, each of size $n$. Each $B[i]$ saves
  the ending position of $\lsus_i^k$, if $\lsus_i^k$ exists;
  {\tt NIL}, otherwise.} 
\KwOut{Array $B$ does not change. Each $A[i] = j$, such that
  $\lsus_j^k$ is the rightmost $\sls_i^k$, if $\sls_i^k$ exists;
 otherwise, $A[i]=\mathtt{NIL}$.}

\bigskip 

\tcc{Find the index of the rightmost existing $k$-mismatch LSUS.}
\lFor{$r=n$ down to $1$}{
  \lIf{$B[r] \neq \nil$}{
    break\;
  }
}

\medskip 

\tcc{Compute the $pred$ array, using the memory space of $A$
  array.  $pred[i]$ is the largest $j$, such that $j<i$ and $|\lsus_j^k|
  < |LSUS_i^k|$, if such $j$ exists; otherwise $pred[i] = \nil$.  If
  $\lsus_i^k$ does not exist, $pred[i] = \nil$ also.  
   From here on, $pred$ and $A$ are the same physical array.}

\medskip 

\lIf{$r<n$}{
  \lFor{$i=r+1\ldots n$}{
    $pred[i] \leftarrow \nil$ \tcp*{Positions that do not have
      $k$-mismatch LSUS}
  }
}

\medskip

$pred[1] \leftarrow \nil$\; 

\For{$i = 2 \ldots r$}{
  $\ell_i \leftarrow B[i] - i + 1$ \tcp*{$|\lsus_i^k|$}
  
  $j \leftarrow i-1$\;

 \lWhile{$pred[j] \neq \nil$ and  $B[j]-j+1 >= \ell_i$}{
    $j \leftarrow pred[j]$;
  }

 \lIf{$B[j]-j+1 < \ell_i$} {
    $pred[i] \leftarrow j$;
  }
  \lElse{
    $pred[i] \leftarrow \nil$;
  }
}

\medskip 

\tcc{ Compute the $t$ array, using the memory space of $A$ array.
  $t[i]$ is the start position of the effective region of $\lsus_i^k$,
  if $\lsus_i^k$ exists; $\nil$, otherwise. From here on, $t$ and $A$
  are the same physical array.}

\medskip 

\For{$i = 1 \ldots r$}{
  \lIf{$pred[i] = \nil$}
  $t[i] \leftarrow i$;
  \lElse{
    $t[i] \leftarrow \max(B[pred[i]] + 1, i)$;}
}

\medskip

\tcc{ Compute the $t^{-1}$ array, using the memory space of array $A$.
$t^{-1}[i]$ is the largest $j$, such that the effective region
  of $\lsus_j^k$ starts on position $i$, if such $j$ exists;
  otherwise, $\nil$. From here on, $t^{-1}$ and $A$ are the same
  physical array.}
  
\medskip

\For{$i = r$ down to $1$}{
  \lIf{$t^{-1}[t[i]] = \nil$}{
    $t^{-1}[t[i]] \leftarrow i$\;
  }
  \lIf{$i < t[i]$}{
    $t^{-1}[i] \leftarrow \nil$\tcp*{Enable us to update
      this place 
      in the future when needed.}
  }
}

\medskip 

\tcc{ Compute $\sls$ array using the memory space of array $A$.
  $\sls[i] = j$, such that $\lsus_j^k$ is the rightmost 
$\sls_i^k$, if $\sls_i^k$ exists;
$\nil$, otherwise. From here on, $\sls$ and $A$ are the same
  physical array.}

\medskip

  $\sls[1] \leftarrow 1$\;

  \lFor{$i = 2 \ldots B[r]$}{
    $\sls[i] \leftarrow \max(\sls[i-1], t^{-1}[i])$;
  }

\end{algorithm}


\begin{algorithm}[h!]
 \caption{Finding  SUS (exact or approximate)}
\label{algo:sus}
\KwIn{Integer arrays $A$ and $B$, each of size $n$. (1) $A[i] = j$,
  such that $\lsus_j^k$ is the rightmost $\sls_i^k$,
 if $\sls_i^k$ exists; otherwise, $A[i]={\tt NIL}$. (2) $B[i]$ is the ending
  position of $\lsus_i^k$, if $\lsus_i^k$ exists; otherwise, $B[i] =
  {\tt NIL}$.} 
\KwOut{Each $(A[i],B[i])$ pair represents the start and ending
  positions of $\sus_i^k$. }

\bigskip 

\For{$i = n$ down to $1$}{
  \If{$B[i] \neq \nil$}{
    $z\leftarrow i+B[i]-1$\tcp*{The rightmost position covered
    by at least one $k$-mismatch LSUS.}
    break;
  }
}

\If{$z<n$}{
  \For{$i=z+1 \ldots n$}{
    $B[i]\leftarrow \nil$ \tcp*{Positions not covered by any
      $k$-mismatch LSUS.}
  }
}

\For{$i = z$ down to $1$}{
  $B[i] \leftarrow B[A[i]]$ \tcp*{The ending position of $\sls_i^k$.}
}

\medskip 

\tcc{By this point, $S\bigl[A[i].. B[i]\bigr] = \sls_i^k$, if
  $\sls_i^k$ exists; otherwise $A[i]=B[i] = \nil$. \newline 
 Note that $\sus_1^k = \sls_1^k = S\bigl[A[1].. B[1]\bigr]$, which must
  be existing and has been computed.\newline 
   Next, we compute $\sus_i^k$ for all $i\geq 2$. }

\medskip

\For{$i=2\ldots n$}{
 \If{$A[i]=B[i]=\nil$}{
    $A[i]\leftarrow A[i-1]$; 
    $B[i] \leftarrow B[i-1]+1$ \tcp*{$\sus_i^k = \sus_{i-1}^kS[i]$}
  }
  \ElseIf{$B[i-1] = i-1$ and $B[i-1]-A[i-1]+2 < B[i]-A[i]+1$}{
    $A[i]\leftarrow A[i-1]$; 
    $B[i] \leftarrow B[i-1]+1$ \tcp*{$\sus_i^k = \sus_{i-1}^kS[i]$}
  }
  
  \tcc{Otherwise, do nothing. $\sus_i^k = \sls_i^k$.}
}
\end{algorithm}


\end{document}